\documentclass{llncs}

\usepackage{epsfig}
\usepackage{amssymb}
\usepackage{amsmath}
\usepackage{amsfonts}
\begin{document}

\frontmatter          
\pagestyle{headings}  
%

%
\mainmatter              
\title{A 3-player protocol preventing persistence in strategic contention with limited feedback}
%
%
\author{George Christodoulou \inst{1} \and Martin Gairing \inst{1} \and Sotiris Nikoletseas \inst{2,3} \and \\ Christoforos Raptopoulos \inst{2,3} \and Paul Spirakis \inst{1,2,3}}
%
%
%
\institute{Department of Computer Science, University of Liverpool, United Kingdom\\
\email{G.Christodoulou@liverpool.ac.uk, gairing@liverpool.ac.uk, P.Spirakis@liverpool.ac.uk}
\and
Computer Engineering and Informatics Department, University of Patras, Greece\\
\and
Computer Technology Institute \& Press ``Diophantus'', Greece\\
\email{nikole@cti.gr, raptopox@ceid.upatras.gr}
}

\maketitle              

\begin{abstract}
In this paper, we study contention resolution protocols from a game-theoretic perspective. In a recent work \cite{ESA2016}, we considered {\em ac\-knowledgment-based} protocols, where a user gets feedback from the channel only when she attempts transmission. In this case she will learn whether her transmission was successful or not. One of the main results of \cite{ESA2016} was that no acknowledgment-based protocol can be in equilibrium. In fact, it seems that many natural acknowledgment-based protocols fail to prevent users from unilaterally switching to persistent protocols that always transmit with probability 1. It is therefore natural to ask how powerful a protocol must be so that it can beat persistent deviators.

In this paper we consider \emph{age-based} protocols, which can be described by a sequence of probabilities of transmitting in each time step. Those probabilities are given beforehand and do not change based on the transmission history. We present a 3-player age-based protocol that can prevent users from unilaterally deviating to a persistent protocol in order to decrease their expected transmission time. It is worth noting that the answer to this question does not follow from the results and proof ideas of \cite{ESA2016}. Our protocol is non-trivial, in the sense that, when all players use it, finite expected transmission time is guaranteed. In fact, we show that this protocol is preferable to any deadline protocol in which, after some fixed time, attempt transmission with probability 1 in every subsequent step.
An advantage of our protocol is that it is very simple to describe, and users
only need a counter to keep track of time.
Whether there exist $n$-player age-based protocols that do not use counters and can prevent persistence is left as an open problem for future research.

\end{abstract}

\section{Introduction}

A fundamental problem in networks is {\em contention resolution} in multiple access channels. In such a setting there are multiple users that want to communicate with each other by sending messages into a multiple access channel (or broadcast channel). The channel is not centrally controlled, so two or more users can transmit their messages at the same time, in which case there is a collision and no transmission is successful. The objective in contention resolution is the design of \emph{distributed protocols} for resolving such conflicts, while simultaneously optimizing some performance measure, like channel utilization or average throughput.

Following the standard assumption in this area, we assume that time is
discrete and messages are broken up into fixed sized packets, which
fit exactly into one time slot. In fact, we consider one of the
simplest possible scenarios where each user only needs to send a
single packet through the channel. Most studies on distributed
contention resolution protocols (see Section~\ref{sec:related-work})
are based on the assumption that users will always follow the
algorithm. In this paper, following~\cite{Fiat07} we drop this
assumption, and we assume that a player will only obey a protocol if
it is in her best interest, given the other players stick to the
protocol. Therefore, we model the situation from a game-theoretic
perspective, i.e. as a stochastic game with the users as selfish
\emph{players}.


One of the main results of Fiat, Mansour, and Nadav~\cite{Fiat07} was the design of an incentive-compatible transmission protocol which guarantees that (with high probability) all players will transmit successfully in time linear in the number of players $n$. The authors assume a \emph{ternary} feedback channel, i.e. each player receives feedback of the form $0/1/2^+$ after each time step, indicating whether zero, one, or more than one transmission was attempted. In a related paper, Christodoulou, Ligett and Pyrga~\cite{CLP14} designed efficient $\epsilon$-equilibrium protocols under a stronger assumption that each player receives as feedback the number of players that attempted transmission; this is called {\em multiplicity} feedback. They also assume non-zero transmission costs, in which case the protocols of~\cite{Fiat07} do not apply.

All of the protocols defined in the above two works belong to the class of {\em full-sensing} protocols~\cite{G03}, in which the channel feedback is broadcasted to all sources. However, in wireless channels, there are situations where full-sensing is not possible because of the \emph{hidden-terminal problem} \cite{TK75}. In a previous work \cite{ESA2016}, we considered {\em acknowledgment-based} protocols, which use a more limited feedback model -- the only feedback that a user gets is whether her transmission was successful or not. A user that does not transmit cannot ``listen'' to the channel and therefore does not get any feedback. In other words, the only information that a user has is the history of her own transmission attempts. Acknowledgment-based protocols have been extensively studied in the literature (see e.g. \cite{G03} and references therein). 

Our main concern in \cite{ESA2016} was the existence of acknowledgment-based protocols that are in equilibrium. For $n=2$ players, we showed that there exists such a protocol, which guarantees finite expected transmission time. Even though the general question for more than 2 players was left open in \cite{ESA2016}, we ruled out that such a protocol can be \emph{age-based}. Age-based protocols are a special case of acknowledgment-based protocols and can be described by a sequence of probabilities (one for each time-step) of transmitting in each time step. Those probabilities are given beforehand and do not change based on the transmission history. The well known ALOHA protocol~\cite{Abr70} is a special age-based protocol, where -- except for the first round -- users always transmit with the same probability. Since an age-based protocol ${\cal P}$ cannot be in equilibrium, it is beneficial for players to deviate from ${\cal P}$ to some other protocol.  In fact, most natural acknowledgment-based protocols fail to prevent users from unilaterally switching to the persistent protocol that always transmits with probability 1. It is therefore natural to ask how powerful a protocol must be with respect to memory (and feedback) in order to be able to prevent persistent deviators.

\subsection{Our Contribution}
\label{sec:contribution}
The question that we consider in this paper is whether there exist age-based protocols that can prevent users from unilaterally deviating to a persistent protocol (in which they attempt a transmission in every step until they successfully transmit) in order to decrease their expected transmission time. In particular, such protocols should be non-trivial, in the sense that using the protocol should guarantee a finite expected transmission time for the users. It is worth noting that the answer to this question does not follow from the results and proof ideas of \cite{ESA2016}. We give a positive answer for the case of 3 players (users), by presenting and analyzing such a protocol (see definition below). In particular, we show that this protocol is preferable to any deadline protocol in which, after some fixed time, attempt transmission with probability 1 in every subsequent step.

Let $c \geq 1$ and $p \in [0,1]$ be constants. We define the protocol ${\cal P} = {\cal P}(c, p)$ as follows: the transmission probability ${\cal P}_t$ at any time $t$ is equal to $p$ if $t = \sum_{j=0}^k \lfloor 2c^j \rfloor$, for some $k = 0, 1, \ldots$ and it is equal to 1 otherwise. The intuition behind this protocol is that with every collision it is increasingly harder for remaining users to successfully transmit, and thus ``aggressive'' protocols are suboptimal.

Our main result is the following:

\begin{theorem} \label{maintheorem}
Assume there are 3 players in the system, two of which use protocol ${\cal P}(1.1, 0.75)$. Then, the third player will prefer using protocol ${\cal P}(1.1, 0.75)$ over any deadline protocol ${\cal D}$.   
\end{theorem}

In addition, we show that the expected transmission time of a fixed player when all players use ${\cal P}(1.1, 0.75)$ is upper bounded by 2759 and is thus finite. We believe that our ideas can be used to give a positive answer also for the case of $n >3$ players, but probably not for too large values of $n$.

An advantage of our protocol is that it is very simple to describe, and users only need a counter to keep track of time. Whether there exist $n$-player age-based protocols using finite memory that can prevent persistence is left as an open problem for future research.

 \subsection{Other Related Work}
 \label{sec:related-work}

 Perhaps the most famous multiple-access communication protocol is the
 (slotted) ALOHA protocol~\cite{Abr70,Rob72}.
Follow-up papers study the efficiency of multiple-access protocols for
packets that are generated by some stochastic process (see
e.g. \cite{355567,310333,903752}), or worst-case scenarios of bursty
inputs \cite{1074023}.

 
The main focus of many contention resolution protocols is on actual
{\em conflict resolution}. In such a scenario, it is assumed that
there are $n$ users in total, and $k$ of them collide. In such an
event, a resolution algorithm is called, which ensures that all the
colliding packets are successfully
transmitted~\cite{capetanakis-a,hayes78,TsM}.
%
There is extensive study on the efficiency of protocols under various
information models (see \cite{G03} for an overview). When
$k$ is known,~\cite{140906} provides an $O(k + \log k \log n)$ {\em
  acknowledgment-based} algorithm, while~\cite{274816} provides a
matching lower bound. For the ternary model,~\cite{GW85} provides a
bound of $\Omega(k(\log n/\log k))$ for all deterministic algorithms.
%

There are various game theoretic models of slotted ALOHA that have
been studied in the literature, apart from the ones mentioned in the
introduction; see for example~\cite{1031826,1154073,Altman05azouzi}.
However, in most of these models only transmission protocols that
always transmit with the same probability are
considered.  
There has been also research on pricing schemes~\cite{1285895} as well
as on cases in which the channel quality changes dynamically with time
and players must choose their transmission levels
accordingly~\cite{Menache08,1288109,AMPP08}. An interesting
game-theoretic model that lies between the contention and congestion
model was studied in \cite{KoutsoupiasP12}; where decisions of {\em
  when} to submit is part of the action space of the players.


\section{Model}
\label{sec:definitions}

\noindent \textbf{Game Structure.} Let $N=\{1, 2, \dots, n\}$ be the set of agents, each one of which has a single packet that he wants to send through a common channel. All players know $n$. We assume time is discretized into slots $t = 1, 2, \ldots$. The players that have not yet successfully transmitted their packet are called \emph{pending} and initially all $n$ players are pending. At any given time slot $t$, a pending player $i$ has two available actions, either to \emph{transmit} his packet or to \emph{remain quiet}. In a \emph{(mixed) strategy}, a player $i$ transmits his packet at time $t$ with some probability that potentially depends on information that $i$ has gained from the channel based on previous transmission
attempts. If exactly one player transmits in a given slot $t$, then
his transmission is \emph{successful}, the successful player exits the
game (i.e. he is no longer pending), and the game continues with the
rest of the players. On the other hand, whenever two or more agents
try to access the channel (i.e. transmit) at the same slot, a
\emph{collision} occurs and their transmissions fail, in which case
the agents remain in the game. Therefore, in case of collision or if
the channel is idle (i.e. no player attempts to transmit) the set of
pending agents remains unchanged. The game continues until all players
have successfully transmitted their packets.

\noindent \textbf{Transmission protocols.}  Let $X_{i, t}$ be the indicator variable that indicates whether player $i$ attempted transmission at time $t$. For any $t \geq 1$, we denote by $\vec{X}_t$ the transmission vector at time $t$, i.e. $\vec{X}_t = (X_{1, t}, X_{2, t}, \ldots, X_{n, t})$. An {\em acknowlegment-based} protocol, uses very limited channel feedback. After each time step $t$, only players that attempted a transmission receive feedback, and the rest get no information. In
fact, the information received by a player $i$ who transmitted during $t$ is whether his transmission was successful (in which case he gets an acknowledgement and exits the game) or whether there was a collision.

Let $\vec{h}_{i, t}$ be the vector of the \emph{personal transmission history} of player $i$ up to time $t$, i.e. $\vec{h}_{i, t} = (X_{i, 1}, X_{i, 2}, \ldots, X_{i, t})$. We also denote by $\vec{h}_t$ the transmission history of all players up to time $t$, i.e. $\vec{h}_t = (\vec{h}_{1, t}, \vec{h}_{2, t}, \ldots, \vec{h}_{n, t})$. In an acknowledgement-based protocol, the actions of player $i$ at time $t$ depend only (a) on his personal history $\vec{h}_{i, t-1}$ and (b) on whether he is pending or not at $t$. A \emph{decision rule} $f_{i, t}$ for a pending player $i$ at time $t$, is a function that maps $\vec{h}_{i, t-1}$ to a probability $\Pr(X_{i, t} = 1 | \vec{h}_{i, t-1})$. For a player $i \in N$, a \emph{(transmission) protocol} $f_i$ is a sequence of decision rules $f_i = \{ f_{i, t}\}_{t \geq 1} = f_{i, 1}, f_{i, 2}, \cdots$.

A transmission protocol is \emph{anonymous} if and only if the
decision rule assigns the same transmission probability to all players
with the same personal history. In particular, for any two players $i
\neq j$ and any $t \geq 0$, if $\vec{h}_{i, t-1} = \vec{h}_{j, t-1}$,
it holds that $f_{i, t}(\vec{h}_{i, t-1}) = f_{j, t}(\vec{h}_{j,
  t-1})$. In this case, we drop the subscript $i$ in the notation,
i.e. we write $f = f_1 = \cdots = f_n$.

We call a protocol $f_i$ for player $i$ \emph{age-based} if and only
if, for any $t \geq 1$, the transmission probability $\Pr(X_{i, t} = 1
| \vec{h}_{i, t-1})$ depends only (a) on time $t$ and (b) on whether
player $i$ is pending or not at $t$. In this case, we will denote the
transmission probability by $p_{i, t} \stackrel{def}{=} \Pr(X_{i, t} =
1 | \vec{h}_{i, t-1}) = f_{i, t}(\vec{h}_{i, t-1})$.

We call a transmission protocol $f_i$ \emph{non-blocking} if and only if, for any $t \geq 1$ and any transition history $\vec{h}_{i, t-1}$, the transmission probability $\Pr(X_{i, t} = 1 | \vec{h}_{i, t-1})$ is always smaller than 1. A protocol $f_i$ for player $i$ is a \emph{deadline protocol with deadline} $t_0 \in \{1, 2, \ldots\}$ if and only if $f_{i, t}(\vec{h}_{i, t-1}) = 1$, for any player $i$, any time slot $t \geq t_0$ and any transmission history $\vec{h}_{i, t-1}$. A \emph{persistent player} is one that uses the deadline protocol with deadline $1$.


\noindent \textbf{Individual utility.} Let ${\vec f} = (f_1, f_2,
\ldots, f_n)$ be such that player $i$ uses protocol $f_i, i \in
N$. For a given transmission sequence $\vec{X}_1, \vec{X}_2, \ldots$,
which is consistent with ${\vec f}$, define the \emph{latency} or
\emph{success time} of agent $i$ as $T_i \stackrel{def}{=}
\inf\{t:X_{i, t} = 1, X_{j, t} = 0, ~ \forall j \neq i\}$. That is,
$T_i$ is the time at which $i$ successfully transmits.  Given
a 
transmission history $\vec{h}_t$, 
the $n$-tuple of protocols ${\vec f}$ induces a probability
distribution over sequences of further transmissions. In that case, we
write $C^{{\vec f}}_i(\vec{h}_t) \stackrel{def}{=}
\mathbb{E}[T_i|\vec{h_t}, {\vec f}] = \mathbb{E}[T_i|\vec{h}_{i, t},
{\vec f}]$ for the expected latency of agent $i$ incurred by a
sequence of transmissions that starts with $\vec{h}_t$ and then
continues based on ${\vec f}$. For anonymous protocols, i.e. when $f_1
= f_2 = \cdots = f_n = f$, we will simply write $C^{f}_i(\vec{h}_t)$
instead\footnote{Abusing notation slightly, we will also write
  $C^{{\vec f}}_i(\vec{h}_0)$ for the \emph{unconditional} expected
  latency of player $i$ induced by ${\vec f}$.
}.

\noindent \textbf{Equilibria.} The objective of every agent is to
minimize her expected latency. We say that ${\vec f} = \{f_1, f_2,
\ldots, f_n \}$ is in \emph{equilibrium} if for any transmission
history $\vec{h}_t$ the agents cannot decrease their expected latency
by unilaterally deviating after $t$; that is, for all agents $i$, for
all time slots $t$, and for all decision rules $f'_i$ for agent $i$,
we have
\begin{displaymath}
C^{{\vec f}}_i(\vec{h}_t) \leq C^{({\vec f}_{-i}, f'_i)}_i(\vec{h}_t),
\end{displaymath}
where $({\vec f}_{-i}, f'_i)$ denotes the protocol
profile\footnote{For an anonymous protocol $f$, we denote by $(f_{-i}, f'_i)$ the
profile where agent $j \neq i$ uses protocol $f$ and agent $i$
uses protocol $f'_i$.} where every
agent $j \neq i$ uses protocol $f_j$ and agent $i$ uses protocol
$f'_i$. 


\section{A 3-player protocol that prevents persistence} \label{sec:persistentwinner}

In this section we prove that there is an anonymous age-based protocol ${\cal P}(c, p)$ for 3 players that has finite expected latency and prevents players from unilaterally switching to any deadline protocol. In what follows, Alice is one of the three players in the system.

For some parameters $c \geq 1$ and $p \in [0,1]$, which will be specified later, we define the protocol ${\cal P} = {\cal P}(c, p)$ as follows:
\begin{equation}
{\cal P}_t = \left\{ 
\begin{array}{ll}
	p, & \quad \textrm{if $t = \sum_{j=0}^k \lfloor 2c^j \rfloor$, for some $k = 0, 1, \ldots$} \\
	1, & \quad \textrm{otherwise.}
\end{array}
\right.
\end{equation}
For $k=0, 1, 2, \ldots$, define the \emph{$k$-th non-trivial transmission time $s_k$} to be the time step on which the decision rule for a pending player using ${\cal P}$ is to transmit with probability $p$. In particular, $s_k \stackrel{def}{=} \sum_{j=0}^k \lfloor 2c^j \rfloor$, by definition of the protocol. 
For technical reasons, we set $s_{k} = 0$, for any $k < 0$.
Furthermore, for $k = 1, 2, \ldots$, define the \emph{$k$-th (non-trivial) inter-transmission time $x_k$} as the time between the $k$-th and $(k-1)$-th non-trivial transmission time, i.e. $x_k \stackrel{def}{=} s_k - s_{k-1} = \lfloor 2c^k \rfloor$. The following elementary result will be useful for the analysis of the protocol. The proof can be found in Appendix \ref{lemma1proof}.

\begin{lemma} \label{dominationlemma}
For any $k, k', j \in \{0, 1, \ldots\}$, such that $k' > k$, and any $c \in [1, 2]$, we have that 
\begin{displaymath}
c^{k'-k-1}(c-1) x_{k+j} \leq x_{k'+j} \leq c^{k'-k-1}(c+1) x_{k+j}.
\end{displaymath}
\end{lemma}

\subsection{Expected latency for a persistent player} \label{section:persistent}

Assume that Alice is a persistent player, i.e. she uses the deadline protocol $g$ with deadline $1$, i.e. $g_t = 1$, for all $t \geq 1$, while both other players use protocol ${\cal P}$. For $n \in \{1, 2, 3\}, k \in \{0, 1, \ldots\}$, let $Y'_{n, k}$ be the additional time after $s_{k-1}$ that Alice needs to successfully transmit when there are $n$ pending players. 
It is evident that Alice will be the first player to successfully transmit, so there will be no need to calculate $\mathbb{E}[Y'_{2, k}]$ or $\mathbb{E}[Y'_{1, k}]$. 


The proof of the following Theorem can be found in Appendix \ref{appendix:persistentexpectedtime}.

\begin{theorem} \label{theorem:persistentexpectedtime}
If $\frac{1}{1 - (1-p)^2} < c \leq 2$, then $\mathbb{E}[Y'_{3, 0}] = \infty$. That is, the expected latency for Alice when she is persistent and both other players use protocol ${\cal P}(c, p)$ is infinity. 
\end{theorem}


\begin{remark} At first glance, the above result may seem surprising. Indeed, let $Z$ denote the number of times that the persistent player has a collision whenever the other players transmit with probability $p$ (i.e. we do not count collissions when the other two players transmit with probability 1, which causes certain collision). It is easy to see that $Z+1$ is a geometric random variable with probability of success $(1-p)^2$. Therefore, $\mathbb{E}[Z+1] = \frac{1}{(1-p)^2}$ is finite! On the other hand, it is not hard to see that the (actual) time $Y'_{3, 0}$ needed for the persistent player to successfully transmit is given by $Y'_{3, 0} = \sum_{j=0}^{Z} \lfloor 2c^j \rfloor$. In particular, $Y'_{3, 0}$ is a strictly convex function of $Z$, for any $c>1$, and so, by Jensen's inequality (see e.g. \cite{RossIntrobook}) $\mathbb{E}[Y'_{3, 0}] > \sum_{j=0}^{\mathbb{E}[Z]} \lfloor 2c^j \rfloor$.
\end{remark}

\subsection{Expected latency when all players use ${\cal P}(c, p)$}

Assume that all three players use protocol ${\cal P}$. For $n \in \{1, 2, 3\}, k \in \{0, 1, \ldots\}$, let $Y_{n, k}$ be the additional time after $s_{k-1}$ that Alice needs to successfully transmit, when there are $n$ pending players. 
The following corollary, which is a direct consequence of Lemma \ref{dominationlemma}, will be useful for our analysis.

\begin{corollary} \label{usualdomination}
For any $n \in \{1, 2, 3\}$, any $k, k' \in \{0, 1, \ldots\}$ with $k' > k$, and $c \in [1, 2]$ we have that 
\begin{displaymath}
c^{k'-k-1}(c-1) \mathbb{E}[Y_{n, k}] \leq \mathbb{E}[Y_{n, k'}] \leq c^{k'-k-1}(c+1) \mathbb{E}[Y_{n, k}].
\end{displaymath}
\end{corollary}

The main purpose of this section is to prove Theorem \ref{theorem:3playersexpectedtime}. To do this, we need to consider $\mathbb{E}[Y_{n, k}]$, for all values of $n \in \{1, 2, 3\}, k \in \{0, 1, \ldots\}$.

\subsubsection{The case $n=1$.} When only Alice is pending, we have
\begin{eqnarray}
\mathbb{E}[Y_{1, k}] & = & \lfloor 2c^k \rfloor p + \left( \sum_{j=k}^{k+1} \lfloor 2c^j \rfloor \right) p(1-p) + \ldots \nonumber \\
& = & \sum_{\ell = k}^{\infty} \left(\left( \sum_{j=k}^{\ell} \lfloor 2c^j \rfloor \right) p(1-p)^{\ell-k} \right) \nonumber \\
& \leq & \sum_{\ell = k}^{\infty} \left(\left( \sum_{j=k}^{\ell} 2c^j \right) p(1-p)^{\ell-k} \right) \nonumber \\
& = & \sum_{\ell = k}^{\infty} \left(\left( 2 \frac{c^{\ell+1} - c^k}{c-1} \right) p(1-p)^{\ell-k} \right) \nonumber \\
& \leq & \frac{2c p}{(c-1) (1-p)^k} \sum_{\ell = k}^{\infty} \left( c^{\ell} (1-p)^{\ell} \right). \label{eq-upperX1k}
\end{eqnarray}
In particular, by the above inequality, we have the following:

\begin{lemma} \label{lemma:1playerexpectedtime}
If $1<c< \frac{1}{1-p}$, then $\mathbb{E}[Y_{1, k}]$ is finite, for any (finite) $k$.
\end{lemma}

\subsubsection{The case $n=2$.} Fix $k_1'>0$ and assume $c \in [1, 2]$ (so that we can apply Corollary \ref{usualdomination}). When two players are pending (i.e. Alice and one other player), we have, for all $i = 0, 1, 2, \ldots, k_1'-1$,
\begin{displaymath}
\mathbb{E}[Y_{2, i}] = \lfloor 2c^i \rfloor + p(1-p) \mathbb{E}[Y_{1, i+1}] + (1-2p(1-p)) \mathbb{E}[Y_{2, i+1}]
\end{displaymath}

Set $\delta = 1-2p(1-p)$. Multiplying the corresponding equation for $\mathbb{E}[Y_{2, i}]$ by $\delta^i$, for each $i = 0, 1, 2, \ldots, k_1'-1$ and adding up, we get 

\begin{equation}
\mathbb{E}[Y_{2, 0}] = \sum_{i=0}^{k_1'-1} \delta^i \lfloor 2c^i \rfloor + p(1-p) \sum_{i=0}^{k_1'-1} \delta^i \mathbb{E}[Y_{1, i+1}] + \delta^{k_1'} \mathbb{E}[Y_{2, k_1'}].
\end{equation}
By the second inequality of Corollary \ref{usualdomination} for $n=2$ and $k=0$, we get
\begin{equation} \label{eq-upperX20}
\mathbb{E}[Y_{2, 0}] \leq \sum_{i=0}^{k_1'-1} \delta^i \lfloor 2c^i \rfloor + p(1-p) \sum_{i=0}^{k_1'-1} \delta^i \mathbb{E}[Y_{1, i+1}] + \delta^{k_1'} c^{k_1'-1}(c+1) \mathbb{E}[Y_{2, 0}]. 
\end{equation}
Observe now that, if we have $1<c<\frac{1}{1-p}$, then, by Lemma \ref{lemma:1playerexpectedtime}, the terms $\sum_{i=0}^{k_1'-1} \delta^i \lfloor 2c^i \rfloor + p(1-p) \sum_{i=0}^{k_1'-1} \delta^i \mathbb{E}[Y_{1, i+1}]$ in the above inequality are finite and strictly positive. Therefore, $\mathbb{E}[Y_{2, 0}]$ (which is also strictly positive), will be finite if, in addition to $c<\frac{1}{1-p}$ and $c \in [1, 2]$, the following inequality holds:
\begin{equation} \label{eq-fork1'}
\delta^{k_1'} c^{k_1'-1}(c+1) < 1.
\end{equation} 
Taking $k_1' \to \infty$ (in fact, given $p, c$, we can choose a minimum, finite value for $k_1'$ so that the above inequality holds, see also Appendix \ref{subsection:expectedpassingtime}), we have that, if $c$ satisfies $c < \frac{1}{1-2p(1-p)}$, and also $c< \frac{1}{1-p}$ (so that $\mathbb{E}[Y_{1, i}]$ is finite for all finite $i$), and $c \in (0, 2]$ (so that we can apply Corollary \ref{usualdomination}), then $\mathbb{E}[Y_{2, 0}]$ is finite. In fact, we can prove the following more general result:

\begin{lemma} \label{lemma:2playersexpectedtime}
If $1< c< \min\left\{\frac{1}{1-p}, \frac{1}{1-2p(1-p)}, 2 \right\}$, then $\mathbb{E}[Y_{2, k}]$ is finite, for any (finite) $k$.
\end{lemma}
\proof By the the above arguments, when $1< c< \min\left\{\frac{1}{1-p}, \frac{1}{1-2p(1-p)}, 2 \right\}$, $\mathbb{E}[Y_{2, 0}]$ is finite. But, by the second inequality of Corollary \ref{usualdomination}, we also have that $\mathbb{E}[Y_{2, k}] \leq c^{k-1}(c+1) \mathbb{E}[Y_{2, 0}]$, which completes the proof. \qed

\subsubsection{The case $n=3$.} Fix $k_2'>0$ and assume $c \in [1, 2]$. When all three players are pending, we have, for all $i = 0, 1, 2, \ldots, k_2'-1$,
\begin{displaymath}
\mathbb{E}[Y_{3, i}] = \lfloor 2c^i \rfloor + 2p(1-p)^2 \mathbb{E}[Y_{2, i+1}] + (1-3p(1-p)^2) \mathbb{E}[Y_{3, i+1}].
\end{displaymath}
Set $\beta = 1-3p(1-p)^2$. Multiplying the corresponding equation for $\mathbb{E}[Y_{3, i}]$ by $\beta^i$, for all $i = 0, 1, 2, \ldots, k_2'-1$ and adding up, we get 
\begin{displaymath}
\mathbb{E}[Y_{3, 0}] = \sum_{i=0}^{k_2'-1} \beta^i \lfloor 2c^i \rfloor + 2p(1-p)^2 \sum_{i=0}^{k_2'-1} \beta^i \mathbb{E}[Y_{2, i+1}] + \beta^{k_2'} \mathbb{E}[Y_{3, k_2'}].
\end{displaymath}
By the second inequality of Corollary \ref{usualdomination} for $n=3$ and $k=0$, we get
\begin{equation} \label{eq-upperX30}
\mathbb{E}[Y_{3, 0}] \leq \sum_{i=0}^{k_2'-1} \beta^i \lfloor 2c^i \rfloor + 2p(1-p)^2 \sum_{i=0}^{k_2'-1} \beta^i \mathbb{E}[Y_{2, i+1}] + \beta^{k_2'} c^{k_2'-1}(c+1) \mathbb{E}[Y_{3, 0}].
\end{equation}
Observe now that, if we have $1< c<\min\left\{\frac{1}{1-p}, \frac{1}{1-2p(1-p)}, 2 \right\}$, then, by Lemma \ref{lemma:2playersexpectedtime}, the terms $\sum_{i=0}^{k_2'-1} \beta^i \lfloor 2c^i \rfloor + 2p(1-p)^2 \sum_{i=0}^{k_2'-1} \beta^i \mathbb{E}[Y_{2, i+1}]$ in the above inequality are finite and strictly positive. Therefore, $\mathbb{E}[Y_{3, 0}]$ (which is also strictly positive), will be finite if, in addition to $1< c<\min\left\{\frac{1}{1-p}, \frac{1}{1-2p(1-p)}, 2 \right\}$, the following inequality holds:
\begin{equation} \label{eq-fork2'}
\beta^{k_2'} c^{k_2'-1}(c+1) < 1.
\end{equation} 
Taking $k_2' \to \infty$ (in fact, given $p, c$, we can choose a minimum, finite value for $k_2'$ so that the above inequality holds, see also Appendix \ref{subsection:expectedpassingtime}), we have that, if $c$ satisfies $c < \frac{1}{1-3p(1-p)^2}$, and also $c< \min\left\{\frac{1}{1-p}, \frac{1}{1-2p(1-p)}, 2 \right\}$ (so that $\mathbb{E}[Y_{2, i}]$ is finite for all finite $i$), then $\mathbb{E}[Y_{3, 0}]$ is finite. Similarly to the proof of Lemma \ref{lemma:2playersexpectedtime}, we can prove the following more general result:

\begin{theorem} \label{theorem:3playersexpectedtime}
If $1< c< \min\left\{\frac{1}{1-p}, \frac{1}{1-2p(1-p)}, \frac{1}{1-3p(1-p)^2}, 2 \right\}$, then $\mathbb{E}[Y_{3, k}]$ is finite, for any (finite) $k$. In particular, the expected latency of Alice when all players (including Alice herself) use protocol ${\cal P}(c, p)$ is finite.
\end{theorem}

\subsection{Feasibility}

We first show that there are values for $p$ and $c$, such that the following inequalities hold at the same time:
\begin{displaymath}
1<c< \min\left\{\frac{1}{1-p}, \frac{1}{1-2p(1-p)}, \frac{1}{1-3p(1-p)^2} \right\}
\end{displaymath}
and 

\begin{displaymath}
\frac{1}{1 - (1-p)^2} < c \leq 2.
\end{displaymath}
By Theorem \ref{theorem:3playersexpectedtime} and Theorem \ref{theorem:persistentexpectedtime}, if all the above inequalities hold, then $\mathbb{E}[Y_{3, 0}]$ is finite, while $\mathbb{E}[Y'_{3, 0}]$ is infinite.

For $p = 3/4$, the above inequalities become: $1<c< \min\{4, 8/5, 64/55\} \approx 1.163$ and $1.066 \approx 16/15 < c \leq 2$. Therefore, selecting $p = 3/4$ and $c = 1.1$, we have an anonymous age-based protocol that has finite expected latency and that prevents players from unilaterally switching to a persistent protocol. In fact we prove a slightly more general result:

\begin{theorem} [restatement of Theorem \ref{maintheorem}]
Assume there are 3 players in the system, two of which use protocol ${\cal P}(1.1, 0.75)$. Then, the third player will prefer using protocol ${\cal P}(1.1, 0.75)$ over any deadline protocol ${\cal D}$.   
\end{theorem}
\begin{proof} Extending the notation used in the previous sections, let $Y^{{\cal D}}_{3, 0}$ (respectively $Y_{3, 0}$) be the time needed for the third player to successfully transmit when she uses protocol ${\cal D}$ (respectively protocol ${\cal P}(1.1, 0.75)$). Furthermore, let $Y'_{3, k}$, $k \in \{0, 1, \ldots\}$, be the additional time after $s_{k-1}$ (i.e. the $(k-1)$-th non-trivial transmission time) that the third player needs to successfully transmit when she uses a deadline protocol with deadline $1$.

Since $c = 1.1$ and $p = 0.75$, we have that $1<c< \min\left\{\frac{1}{1-p}, \frac{1}{1-2p(1-p)}, \frac{1}{1-3p(1-p)^2} \right\}$ and $\frac{1}{1 - (1-p)^2} < c \leq 2$. Therefore, by Theorems \ref{theorem:persistentexpectedtime} and \ref{theorem:3playersexpectedtime}, we have that $\mathbb{E}[Y_{3, 0}]$ is finite and $\mathbb{E}[Y'_{3, 0}] = \infty$, which means that the third player prefers using ${\cal P}(1.1, 0.75)$ over a deadline protocol with deadline $1$.

We now prove that the third player prefers using ${\cal P}(1.1, 0.75)$ over any deadline protocol ${\cal D}$ with deadline $t_0 = t_0({\cal D})$ as well. Let ${\cal E}$ be the event that none of the first two players has successfully transmitted before $t_0$. Let also $\xi = \xi(t_0)$ be the number of times $t$ such that ${\cal P}(1.1, 0.75)_t = p = 0.75$ (i.e. the protocol ${\cal P}(1.1, 0.75)$ suggests transmitting with probability less than 1) before time $t_0$ (i.e. $\xi(t_0)$ is the number of non-trivial transmissions before $t_0$). We can see that 
\begin{displaymath}
\Pr({\cal E}) \geq (1-2p(1-p))^{\xi}.
\end{displaymath} 
In fact, this lower bound is quite crude, since it does not take into account the third player, so the probability that one of the first two players succesffully transmits during a non-trivial transmission time step when both are pending is $2p(1-p)$. We now have the following:
\begin{eqnarray}
\mathbb{E}\left[Y^{{\cal D}}_{3, 0} \right] & = & \sum_{t=0}^{\infty} t \Pr\left(Y^{{\cal D}}_{3, 0} = t \right) \nonumber \\
& \geq & \Pr({\cal E}) \sum_{t=\tau}^{\infty} t \Pr\left(\ Y^{{\cal D}}_{3, 0} = t | {\cal E} \right) = \Pr({\cal E}) \sum_{t=\tau}^{\infty} t \Pr(Y'_{3, \xi} = t) \nonumber \\
& \geq & \Pr({\cal E}) \sum_{t=0}^{\infty} t \Pr(Y'_{3, \xi} = t) - t_0^2 = \Pr({\cal E}) \mathbb{E}[Y'_{3, \xi}] -t_0^2 \nonumber \\
& \geq & \Pr({\cal E}) c^{\xi-1}(c-1) \mathbb{E}[Y'_{3, 0}] -t_0^2 = \infty \nonumber
\end{eqnarray}
where in the last inequality we used the first inequality of Corollary \ref{usualdomination}. Therefore, the third player prefers using ${\cal P}(1.1, 0.75)$ over ${\cal D}$ as well. Since ${\cal D}$ is arbitrary, the proof is complete. 
\qed
\end{proof}

In Appendix \ref{subsection:expectedpassingtime}, we show that when all three players use the protocol ${\cal P}(1.1, 0.75)$, the expected latency of a fixed player is upper bounded by 2759. It is worth noting that a naive protocol where each player transmits with constant probability, say $\frac{1}{3}$, at any time $t$, has a better expected latency than that of ${\cal P}(c, p)$, but on the other hand it does not prevent players from unilaterally switching to some deadline protocol.

\bibliographystyle{plain}
\bibliography{aloha}

\newpage

\appendix

\section{Proof of Lemma \ref{dominationlemma}} \label{lemma1proof}

By definition of the protocol ${\cal P}(c, p)$, we have that $x_{k} = \lfloor 2c^k \rfloor$, for any $k \in \{0, 1, \ldots\}$. 
Let $A_1, A_2, A_3 \in \mathbb{N}$ and $a_1, a_2, a_3 \in [0, 1)$, such that
\begin{eqnarray}
2c^{k'+j} & = & A_1+a_1, \nonumber \\
2c^{k+j} & = & A_2+a_2, \nonumber \\
c^{k'-k} & = & A_3+a_3. \nonumber 
\end{eqnarray}
We then have that $A_1+ a_1 = (A_3 + a_3)A_2 + (A_3+a_3)a_2$. Therefore,
\begin{eqnarray}
c^{k'-k-1}(c+1) x_{k+j} & = & (A_3+a_3) \left(1+\frac{1}{c} \right) A_2 \nonumber \\
& = & (A_3+a_3) A_2 + \frac{1}{c} (A_3+a_3) A_2 \nonumber \\
& = & A_1+a_1 - (A_3+a_3)a_2 + \frac{1}{c} (A_3+a_3) A_2 \nonumber \\
& = & A_1+a_1 + (A_3+a_3)\left(\frac{1}{c} A_2 -a_2 \right) \nonumber \\
& \geq & A_1+a_1 + (A_3+a_3)\left(\frac{2}{c} -a_2 \right) \geq A_1 + a_1, \nonumber
\end{eqnarray}
where in the first inequality we used the fact that $A_2 \geq 2$ and $a_2 < 1$. This completes the proof of the first inequality of the Lemma. 

The proof for the second inequality of the Lemma is similar. In particular,
\begin{eqnarray}
c^{k'-k-1}(c-1) x_{k+j} & = & (A_3+a_3) \left(1-\frac{1}{c} \right) A_2 \nonumber \\
& = & (A_3+a_3) A_2 - \frac{1}{c} (A_3+a_3) A_2 \nonumber \\
& = & A_1+a_1 - (A_3+a_3)a_2 - \frac{1}{c} (A_3+a_3) A_2 \nonumber \\
& = & A_1+a_1 - (A_3+a_3)\left(\frac{1}{c} A_2 +a_2 \right) \nonumber \\
& \leq & A_1+a_1 - \left(\frac{2}{c} +a_2 \right) \leq A_1 \nonumber
\end{eqnarray}
where in the above inequality we used the fact that, since $k'>k$ and $c \in [1,2]$, we have that $A_3+a_3 \geq 1$ and $A_2 \geq 2$. This completes the proof.

\section{Proof of Theorem \ref{theorem:persistentexpectedtime}} \label{appendix:persistentexpectedtime}

The following corollary, which is a direct consequence of Lemma \ref{dominationlemma}, will be useful for our analysis.

\begin{corollary} \label{persistentdomination}
For any $k, k' \in \{0, 1, \ldots\}$ with $k' > k$, and $c \in [1, 2]$ we have that 
\begin{displaymath}
c^{k'-k-1}(c-1) \mathbb{E}[Y'_{3, k}] \leq \mathbb{E}[Y'_{3, k'}] \leq c^{k'-k-1}(c+1) \mathbb{E}[Y'_{3, k}]. 
\end{displaymath}
\end{corollary}

For the proof of Theorem \ref{theorem:persistentexpectedtime}, fix $k'>0$. By well known properties of expectation, we have, for all $i = 0, 1, \ldots, k'-1$, 
\begin{displaymath}
\mathbb{E}[Y'_{3, i}] = \lfloor 2c^i \rfloor + (1 - (1-p)^2) \mathbb{E}[Y'_{3, i+1}].
\end{displaymath}
Set $\gamma = 1 - (1-p)^2$. Multiplying the corresponding equation for $\mathbb{E}[Y'_{3, i}]$ by $\gamma^i$, for each $i = 0, 1, 2, \ldots, k'-1$ and adding up, we get 
\begin{displaymath}
\mathbb{E}[Y'_{3, 0}] = \sum_{i=0}^{k'-1} \gamma^i \lfloor 2c^i \rfloor + \gamma^{k'} \mathbb{E}[Y'_{3, k'}].
\end{displaymath}
By the first inequality of Corollary \ref{persistentdomination} for $k=0$, we get
\begin{displaymath}
\mathbb{E}[Y'_{3, 0}] \geq \sum_{i=0}^{k'-1} \gamma^i \lfloor 2c^i \rfloor + \gamma^{k'} c^{k'-1}(c-1) \mathbb{E}[Y'_{3, 0}].
\end{displaymath}
Observe now that, since $\sum_{i=0}^{k'-1} \gamma^i \lfloor 2c^i \rfloor$ in the above equation is strictly positive, $\mathbb{E}[X'_{3, 0}]$ (which is also strictly positive), will be $\infty$ if the following inequality holds:
\begin{displaymath}
\gamma^{k'} c^{k'-1}(c-1) \geq 1.
\end{displaymath} 
Taking $k' \to \infty$, we have that, for any constant $c$ such that $c > \frac{1}{1 - (1-p)^2}$, we have that $\mathbb{E}[Y'_{3, 0}] = \infty$, which completes the proof of Theorem \ref{theorem:persistentexpectedtime}.

\section{An upper bound on the expected latency when all players use ${\cal P}(c, p)$} \label{subsection:expectedpassingtime}

In this section, we provide an upper bound on the expected latency of a fixed player when all three players use the protocol ${\cal P}(c, p)$, for some $c, p$ such that the conditions of Theorem \ref{theorem:3playersexpectedtime} and Theorem \ref{theorem:persistentexpectedtime} are satisfied, i.e. $1<c< \min\left\{\frac{1}{1-p}, \frac{1}{1-2p(1-p)}, \frac{1}{1-3p(1-p)^2} \right\}$ and $\frac{1}{1 - (1-p)^2} < c \leq 2$.

By equation (\ref{eq-upperX1k}) and using the fact that, by assumption, $c(1-p)<1$, we have that, for any $k \geq 0$,
\begin{eqnarray}
\mathbb{E}[Y_{1, k}] & \leq & \frac{2c p}{(c-1) (1-p)^k} \sum_{\ell = k}^{\infty} \left( c^{\ell} (1-p)^{\ell} \right) \nonumber \\
& = & \frac{2c p}{(c-1) (1-p)^k} \frac{c^k (1-p)^k}{1 - c(1-p)} \nonumber \\
& = & \frac{2cp}{(c-1) (1 - c(1-p))} \left( \frac{c}{1-p}\right)^k. \label{eq-realupperX1k}
\end{eqnarray}

By equation (\ref{eq-upperX20}), we also have
\begin{eqnarray}
&& \left(1 - \delta^{k_1'} c^{k_1'-1}(c+1) \right) \mathbb{E}[Y_{2, 0}] \nonumber \\
&& \qquad \leq \sum_{i=0}^{k_1'-1} 2 \delta^i c^i + p(1-p) \sum_{i=0}^{k_1'-1} \delta^i \mathbb{E}[Y_{1, i+1}] \nonumber \\
&& \qquad \leq \sum_{i=0}^{k_1'-1} 2 \delta^i c^i + \frac{2c^2p^2}{(c-1) (1 - c(1-p))} \sum_{i=0}^{k_1'-1} \left( \frac{\delta c}{1-p}\right)^i \nonumber \\
&& \qquad \leq 2 \frac{(\delta c)^{k_1'}-1}{\delta c-1} + \frac{2c^2p^2}{(c-1) (1 - c(1-p))} \frac{1 - \left(\frac{\delta c}{1-p}\right)^{k_1'}}{1 - \frac{\delta c}{1-p}}, \nonumber
\end{eqnarray}
where in the second equality we used the upper bound from equation (\ref{eq-realupperX1k}). Rearranging, we have
\begin{eqnarray}
\mathbb{E}[Y_{2, 0}] & \leq & \frac{2\frac{1 - (\delta c)^{k_1'}}{1 - \delta c} + \frac{2c^2p^2}{(c-1) (1 - c(1-p))} \frac{1 - \left(\frac{\delta c}{1-p}\right)^{k_1'}}{1 - \frac{\delta c}{1-p}}}{1 - \delta^{k_1'} c^{k_1'-1}(c+1) } \nonumber \\
& \stackrel{def}{=} & \Delta(c, p, k_1'). \nonumber
\end{eqnarray}
By the second inequality of Corollary \ref{usualdomination}, we then have, for any $k \geq 1$,
\begin{equation} \label{eq-realupperX2k}
\mathbb{E}[Y_{2, k}] \leq c^{k-1}(c+1) \mathbb{E}[Y_{2, 0}] \leq \Delta(c, p, k_1') c^{k-1}(c+1).
\end{equation}

As mentioned earlier, $k_1'$ must be large enough, so that equation (\ref{eq-fork1'}) is satisfied. For $c = 1.1$ and $p = 0.75$ (i.e. when all players play protocol ${\cal P}(1.1, 0.75)$), we need $k_1' \geq 2$. In particular, in the case of ${\cal P}(1.1, 0.75)$, taking $k_1'=2$, we have $\Delta(1.1, 0.75, 2) \approx 755.56$, and so the above inequality becomes $\mathbb{E}[Y_{2, k}] \leq 756 c^{k-1}(c+1)$.

Similarly, by equation (\ref{eq-upperX30}), we have
\begin{eqnarray}
&& (1 - \beta^{k_2'} c^{k_2'-1}(c+1)) \mathbb{E}[Y_{3, 0}] \nonumber \\
&& \qquad \leq \sum_{i=0}^{k_2'-1} 2 \beta^i c^i + 2p(1-p)^2 \sum_{i=0}^{k_2'-1} \beta^i \mathbb{E}[Y_{2, i+1}] \nonumber \\
&& \qquad \leq \sum_{i=0}^{k_2'-1} 2 \beta^i c^i + 2p(1-p)^2 (c+1) \Delta(c, p, k_1') \sum_{i=0}^{k_2'-1} (\beta c)^{i} \nonumber 
\end{eqnarray}
where in the second equality we used the upper bound from equation (\ref{eq-realupperX2k}). Taking $k_2' \to \infty$ and using the fact that, by assumption, $\beta c < 1$, we have 
\begin{displaymath}
\mathbb{E}[Y_{3, 0}] \leq \left(2 + 2p(1-p)^2 (c+1) \Delta(c, p, k_1')\right) \frac{1}{1- \beta c}.
\end{displaymath}
Setting now $c = 1.1, p=0.75$ (i.e. all players play protocol ${\cal P}(1.1, 0.75)$) and $k_1'=2$ (in which case $\Delta(1.1, 0.75, 2) \leq 756$), we have that $\mathbb{E}[Y_{3, 0}] \leq 2759$.


%
%
%


\end{document}